\documentclass[conference,letterpaper]{IEEEtran}

\usepackage[utf8]{inputenc} 
\usepackage[T1]{fontenc}
\usepackage{url}
\usepackage{ifthen}
\usepackage{cite}
\usepackage{amssymb, amsfonts}
\usepackage[cmex10]{amsmath}
\usepackage{bm}
\usepackage{graphicx} 
\usepackage{float}   
\usepackage{multirow} 
\usepackage{booktabs} 
\usepackage{xcolor}
\usepackage{mathtools}
\usepackage{geometry}
\usepackage{epstopdf} 
\newtheorem{theorem}{\textbf{Theorem}}
\newtheorem{lemma}[theorem]{\textbf{Lemma}}

\newenvironment{proof}
  {\par\noindent\hspace{1em}\itshape Proof:\ \normalfont}
  {\par\nobreak\hfill$\square$\par}

\makeatletter

\newcommand{\Rmnum}[1]{\expandafter\@slowromancap\romannumeral #1@}
\makeatother
\begin{document}
\newgeometry{left=0.64in,right=0.625in,top=0.73in,bottom=1.05in}
\title{Successive Cancellation List Decoding of Extended Reed-Solomon Codes} 
\author{
	\IEEEauthorblockN{
    Xiaoqian Ye\IEEEauthorrefmark{2},
		Jingyu Lin\IEEEauthorrefmark{2}\IEEEauthorrefmark{3},
    Junjie Huang\IEEEauthorrefmark{4},
		Li Chen\IEEEauthorrefmark{2}\IEEEauthorrefmark{3},
    and Chang-An Zhao\IEEEauthorrefmark{4}\\
		\IEEEauthorblockA{\IEEEauthorrefmark{2}
			School of Electronics and Information Technology, Sun Yat-sen University, Guangzhou 510006, P. R. China
		}
    \IEEEauthorblockA{\IEEEauthorrefmark{3}
      Guangdong Province Key Laboratory of Information Security Technology, Guangzhou 510006, P. R. China
    }
    \IEEEauthorblockA{\IEEEauthorrefmark{4}
			School of Mathematics, Sun Yat-sen University, Guangzhou 510275, P. R. China
		}
    \IEEEauthorblockA{Email: \{yexq26, linjy228, huangjj76\}@mail2.sysu.edu.cn, \{chenli55, zhaochan3\}@mail.sysu.edu.cn}
}}
\maketitle
\begin{abstract}
THIS PAPER IS ELIGIBLE FOR THE STUDENT PAPER AWARD.
Reed-Solomon (RS) codes are an important class of non-binary error-correction codes. They are particularly competent in correcting burst errors, being widely applied in modern communications and data storage systems. This also thanks to their distance property of reaching the Singleton bound, being the maximum distance separable (MDS) codes. This paper proposes a new list decoding for extended RS (eRS) codes defined over a finite field of characteristic two, i.e., $\mathbb{F}_{2^n}$. It is developed based on transforming an eRS code into $n$ binary polar codes. Consequently, it can be decoded by the successive cancellation (SC) decoding and further their list decoding, i.e., the SCL decoding. A pre-transformed matrix is required for re-interpretating the eRS codes, which also determines their SC and SCL decoding performances. Its column linear independence property is studied, leading to theoretical characterization of their SC decoding performance. Our proposed decoding and analysis are validated numerically. 
\end{abstract}


\section{Introduction}
Reed-Solomon (RS) codes are an important class of non-binary error-correction codes, which are particularly competent in correcting burst errors. They have been widely applied in modern communications and data storage systems. Algebraic decoding of RS codes can be categorized into the industralized syndrome based decoding and the more recent curve-fitting based decoding. The celebrated Berlekamp-Massey (BM) algorithm \cite{BM_algorithm}\cite{shift_register_BCH} is a classic example of the former, which can correct errors up to half the code's minimum Hamming distance. For an $(N, K)$ RS code, where $N$ and $K$ are the length and dimension of the code, respectively, the BM algorithm exhibits a decoding complexity of $\mathcal{O}(N^2)$. By further utilizing soft received information for identifying $\gamma$ unreliable received symbols to be flipped, the Chase-BM decoding\cite{chase} yields improved performances by executing $2^\gamma$ decoding events, exhibiting a complexity of $\mathcal{O}(2^\gamma N^2)$. The curve-fitting based decoding is substantiated by the  Guruswami-Sudan (GS) algorithm\cite{GS_algorithm}, which is also a list decoding algorithm in nature. It can correct errors beyond the half distance bound, especially for low rate codes. With an interpolation multiplicity $m$ and a maximum decoding output list size $L$, the GS algorithm exhibits a decoding complexity of $\mathcal{O}(Lm^4N^2)$. Further performance improvement can be achieved by the K\"{o}tter-Vardy (KV) algorithm\cite{KV_algorithm} that transforms the reliability information into interpolation multiplicities. 
It exhibits a decoding complexity of $\mathcal{O}(N^2L^5)$\cite{KV_compliexity}. It can be seen that the existing soft-decision RS decoding approaches remain complex. A more efficient soft decoding approach for RS codes is yet to be developed.

Meanwhile, modern codes have advanced immensely in terms of achieving Shannon capacity, which can be demonstrated by the success of polar codes\cite{arikan2009polar}. Aiming to increase the minimum Hamming distance of polar codes, it has been shown that they can be constructed with dynamic frozen bits by the use of parity-check matrix of an extended BCH (eBCH) code \cite{trifonov_dynamic}. Inversely, a recent research of \cite{binary_transformation} has shown that a binary linear block code (B-LBC) can be interpreted as a polar code with dynamic frozen bits. Consequently, the B-LBC can be decoded by the successive cancellation (SC) decoding algorithm\cite{arikan2009polar} or its performance improving SC list (SCL) decoding algorithm\cite{SCL_decoding}. For this transformation, distribution of the B-LBC information bits over the polar code paradigm is essential for determining the soft decoding performance. This is determined by the permutation matrix that is required for the transform. 

Extending the above soft decoding approaches into handling non-binary linear block codes (NB-LBCs), it was first shown in \cite{peter_SC_4_RS_codes} that RS codes can be similarly transformed into non-binary polar codes with dynamic frozen symbols. Subsequently, SC decoding and sequential decoding of RS codes were proposed. With a maximum list size of $L$, the sequential decoding has a complexity of $\mathcal{O}(LN^2\log_2^2N)$. Alternatively, a more recent research of \cite{SCL_NonBinary} showed that an NB-LBC defined over a finite field of characteristic two, i.e., $\mathbb{F}_{2^n}$ can be transformed into $n$ polar codes, where $n$ is the order of the field. Subsequently, SC and SCL decoding can also be performed for the $n$ binary polar codes, constituting a soft decoding of the NB-LBC.

This paper focuses on SC and SCL decoding of extended RS (eRS) codes and their performance analysis. An eRS code defined over $\mathbb{F}_{2^n}$ can be decoded by the SC and SCL decoding of $n$ binary polar codes. This reinterpretation is bridged by a pre-transformed matrix. In order to characterize the theoretical SC decoding performance, linear independence property of columns in the pre-transformed matrix is studied. It has been aware that performance of SC and SCL decoding partly depends on the richness of \textit{apriori} information, which can be represented by the number of frozen symbols before estimating an information symbol. Hence, our analysis unveils that limited by the design of the pre-transformed matrix, the SC and SCL decoding performances degrade as the eRS codeword length increases. Our simulation results validate the analysis and show the proposed SCL decoding outperforms both the KV and Chase-BM decoding for length-32 eRS codes. 

\textbf{Notation:} Let $\mathbb{F}_{2^n}=\left\{\sigma_0, \sigma_1, \cdots, \sigma_{2^n-1}\right\}$ denote a finite field with characteristic two and order $n$, where $\sigma_0$ designates the zero element. Let $\alpha$ and $p(X)$ denote its primitive element and primitive polynomial, respectively. Consequently, given $U\in\mathbb{F}_{2^{n}}$, it can be represented as $\sum_{j=0}^{n-1}u_j\alpha^{j}$, where $u_j\in\mathbb{F}_2$. Vector $(u_0, u_1, \cdots, u_{n-1})$ is the binary composition of $U$. In this paper, it is also let $U\left[j\right]=u_j$, where $j=0, 1, \cdots, n-1$.

\section{Preliminaries}
This section introduces the foundations of eRS code and polar code.
\subsection{eRS Codes}
Let $\mathcal{C}$ denote an ($N=2^n, K$) eRS code defined over $\mathbb{F}_{2^n}$. Given a message 
$\bm{F}=(F_0, F_1, \cdots, F_{K-1}) \in \mathbb{F}_{2^n}^K$
, it can be presented as
\begin{equation}
  F(x)=F_0+F_1x+\cdots+F_{K-1}x^{K-1}.
\end{equation}
Codeword $\bm{C}$ can be generated by 
\begin{equation}
  \bm{C}=(F(\sigma_1), F(\sigma_2), \cdots, F(\sigma_{2^n-1}), F(\sigma_0)), \label{assignment_coding}
\end{equation}
where $\bm{C}\in\mathbb{F}_{2^n}^{N}$ and locators $\sigma_0, \sigma_1, \cdots, \sigma_{2^n-1}$ are the $N$ distinct elements of $\mathbb{F}_{2^n}$. Note that an RS code defined over $\mathbb{F}_{2^n}$ is often encoded with a length of $2^n-1$, for which $\sigma_0$ is not considered as a locator. In this work, it is considered as $C_{N-1}=F(\sigma_0)$ so that the codeword length reaches $2^n$, legalizing their transformation into polar codes that should have a length in the orders of two. Meanwhile, since $F(\sigma_0)=\sum_{i=0}^{N-2}C_i$, it is also a parity symbol. 

Let $\bm{{\rm G}}$ denote the generator matrix of $\mathcal{C}$, and $\bm{{\rm G}}\in\mathbb{F}_{2^n}^{K\times N}$. Codewords of $\mathcal{C}$ can also be generated by $\bm{C}=\bm{F}\bm{{\rm G}}$. Since RS and eRS codes are maximum distance separable (MDS) codes, any $K$ columns of $\bm{{\rm G}}$ are linearly independent\cite{eRS_MDS_book}\cite{eRS_MDS}.

\subsection{Polar Codes} 
Let us consider a binary polar code also of length $N$ and dimension $K$. An input vector $\bm{u}=(u_0, u_1, \cdots, u_{N-1})\in\mathbb{F}_2^{N}$ consists of $K$ information bits and $N-K$ redundant bits which are also called frozen bits. Based on channel polarization theorem\cite{arikan2009polar}, the $K$ most reliable subchannels are chosen to transmit information bits, while the remaining $N-K$ subchannels are utilized to transmit the frozen bits. The process of categorizing the subchannels is called rate profiling. Generator matrix of the polar code can be defined as $\bm{{\rm G}}_{\rm p} = \bm{{\rm F}}^{\otimes n}$,
where $\mathbf{F}=\left[(1, 0), (1, 1)\right]^{\rm T}$ is the Ar\i kan kernel and $\otimes n$ denotes the $n$-fold Kronecker product. Note that $\bm{{\rm G}}_{\rm p}\in\mathbb{F}_2^{N\times N}$. The polar codeword $\bm{c}$ can be generated by 
\begin{equation}
  \bm{c} = \bm{u}\bm{{\rm G}}_{\rm p}.
\end{equation} 

\section{SCL Decoding of eRS Codes}
This section introduces the transformation from an $(N, K)$ eRS code to $n$ binary polar codes, and its SC and SCL decoding. The SCL decoding complexity is also characterized for eRS codes.
\subsection{Transformation from eRS Codes to Polar Codes}
For an $(N, K)$ eRS code $\mathcal{C}$, the binary composition of any codeword $\bm{C} \in \mathcal{C}$ can be represented as the concatenation of $n$ permuted binary polar codewords\cite{SCL_NonBinary}. Let $\bm{{\rm P}}\in \mathbb{F}_{2}^{N\times N}$ denote a permutation matrix. The eRS codebook $\mathcal{C}$ can be defined as
\begin{subequations}
  \begin{align}
    \mathcal{C} &\triangleq \{\bm{C}=\bm{F}\bm{{\rm G}}\,|\, \forall\bm{F}\in \mathbb{F}_{2^n}^{K}\} \label{eq:a}\\
    &=\{\bm{C}=\bm{F}(\bm{{\rm G}}\bm{{\rm P}}^{-1}\bm{{\rm G}}_{\rm p}^{-1})\bm{{\rm G}}_{\rm p}\bm{{\rm P}}\,| \,\forall\bm{F}\in \mathbb{F}_{2^n}^{K}\} \label{eq:b}\\
    &=\{\bm{C}=\bm{U}\bm{{\rm G}}_{\rm p}\bm{{\rm P}}\,|\,\bm{U}=\bm{F}\bm{{\rm G}}\bm{{\rm P}}^{-1}\bm{{\rm G}}_{\rm p}^{-1}, \text{and } \forall\,\bm{F}\in\mathbb{F}_{2^n}^{K}\},\label{eq:c}
  \end{align}
\end{subequations}
where $\bm{U}=(U_0, U_1, \cdots, U_{N-1})\in\mathbb{F}_{2^n}^N$ can be seen as the non-binary input vector to a polar encoder.
Let 
\begin{equation}
  \begin{aligned}
    \hspace{-0.2em}\bm{u} = (u_{0, 0}, u_{0, 1}, \cdots&, u_{0, n-1}, u_{1, 0}, u_{1, 1}, \cdots, u_{1, n-1}, \\
    & \cdots, u_{N-1, 0}, u_{N-1, 1}, \cdots, u_{N-1, n-1}) \label{all_element_u}
  \end{aligned}
\end{equation}
denote the binary composition of $\bm{U}$, where $\bm{u} \in \mathbb{F}_{2}^{N\cdot n}$. Entries of $\bm{u}$ can be partitioned into $n$ groups, each of which consists of $N$ bits, i.e., 
\begin{equation}
  \bm{u}_j = (u_{0, j}, u_{1, j}, \cdots, u_{N-1, j}), \label{part_ele_u}
\end{equation}
for $j=0, 1, \cdots, n-1$. Correspondingly, let
\begin{equation}
  \begin{aligned}\label{codeword_composition}
    \bm{c} = (c_{0, 0}, c_{0, 1}, \cdots&, c_{0, n-1}, c_{1, 0}, c_{1, 1}, \cdots, c_{1, n-1}, \\
    & \cdots, c_{N-1, 0}, c_{N-1, 1}, \cdots, c_{N-1, n-1})
  \end{aligned}
\end{equation}
denote the binary composition of $\bm{C}$, and $\bm{c}\in \mathbb{F}_{2}^{N\cdot n}$. Similarly, $\bm{c}$ can also be partitioned into $n$ groups as
\begin{equation}
  \bm{c}_j = (c_{0, j}, c_{1, j}, \cdots, c_{N-1, j}),
\end{equation}
for $j=0, 1, \cdots, n-1$. Since $\bm{{\rm G}}_{\rm p}$, $\bm{{\rm P}}\in\mathbb{F}_2^{N\times N}$, the encoding operation in \eqref{eq:c} implies it can be performed over $\mathbb{F}_2$. In particular, $\bm{c}_j$ is a permuted binary polar codeword of $\bm{u}_j$, i.e., 
\begin{equation}
  \bm{c}_j=\bm{u}_j\bm{{\rm G}}_{\rm p}\bm{{\rm P}}. \label{composition_encoding}
\end{equation}
Therefore, the binary composition of an eRS codeword $\bm{C}$ can be interpreted as a concatenation of $n$ permuted binary polar codewords. 
To determine the distribution of information and frozen symbols, Gaussian elimination (GE) is further performed on $\bm{{\rm G}}\bm{{\rm P}}^{-1}\bm{{\rm G}}_{\rm p}^{-1}$. Based on the fact that $\bm{{\rm G}}_{\rm p}=\bm{{\rm G}}_{\rm p}^{-1}$, the non-binary pre-transformed matrix
\begin{equation}
  \mathbf{M}=\mathbf{E}\bm{{\rm G}}\bm{{\rm P}}^{-1}\bm{{\rm G}}_{\rm p} \label{pre-transformed_matrix}
\end{equation}
can be obtained, where $\mathbf{E}\in \mathbb{F}_{2^n}^{K\times K}$ is a full rank row elimination matrix. Matrix $\mathbf{M}\in\mathbb{F}_{2^n}^{K\times N}$ is in a row reduced echelon form. Hence, based on \eqref{eq:b} and \eqref{pre-transformed_matrix}, the eRS codebook can be rewritten as 
\begin{equation}\label{transformation_with_E}
    \mathcal{C} \triangleq \{\bm{C}=\bm{F}'\mathbf{M}\bm{{\rm G}}_{\rm p}\bm{{\rm P}}\,|\, \bm{F}'=\bm{F}\bm{{\rm E}}^{-1}, \text{and } \forall\,\bm{F}\in \mathbb{F}_{2^n}^{K}\},
\end{equation}
and the non-binary input vector can also be represented as 
\begin{equation}
  \bm{U}=\bm{F}'\mathbf{M}, \label{non_binary_input_vector}
\end{equation}
where $\bm{F}'$ is the transformed message. 
Fig. \ref{encoding} summarizes the transformation from an eRS codeword into $n$ permuted binary polar codewords. 

Pre-transformed matrix $\mathbf{M}$ implies the distribution of information and frozen symbols in $\bm{U}$. Note that the binary polar input $\bm{u}_0, \bm{u}_1, \cdots, \bm{u}_{n-1}$ share an identical distribution. Let $\mathcal{A}$ denote the index set of the pivot columns of $\mathbf{M}$, which is also the index set of the information symbols in $\bm{U}$. Hence, $\mathcal{A}^{c}=\{0, 1,\cdots, N-1\}\setminus\mathcal{A}$ denotes the index set of the frozen symbols in $\bm{U}$. Subsequently, we define $\mathbf{M}^{\mathcal{A}}$ and $\mathbf{M}^{\mathcal{A}^c}$ as the submatrices of $\mathbf{M}$, which are constituted by columns indexed in $\mathcal{A}$ and $\mathcal{A}^{c}$, respectively. Since matrix $\bm{{\rm M}}$ is in a row reduced echelon form, the frozen symbol with index $i\in\mathcal{A}^{c}$ is a linear combination of information symbols with indices smaller than $i$, i.e., 
\begin{equation}
  U_i=\sum_{t=0}^{\tau_i}F'_t\cdot \mathbf{M}_{t, i}, \label{dynamic_frozen_symbol}
\end{equation}
where $\tau_i = |\mathcal{A}\cap \{0, 1, \cdots, i-1\}|$. Based on the pre-transformed matrix, the frozen symbols of $\bm{U}$ can be further categorized into the \textit{static} ones and the \textit{dynamic} ones. They correspond to the all-zero columns and the remaining columns in $\mathbf{M}^{\mathcal{A}^{c}}$, respectively. Therefore, permutation matrix $\bm{{\rm P}}$ can adjust the pivot columns of $\mathbf{M}$, and further influences the distribution of information symbols over $\bm{U}$. 

\begin{figure}[h]
  \vspace{-4mm}
	\centerline{\includegraphics[scale=0.61]{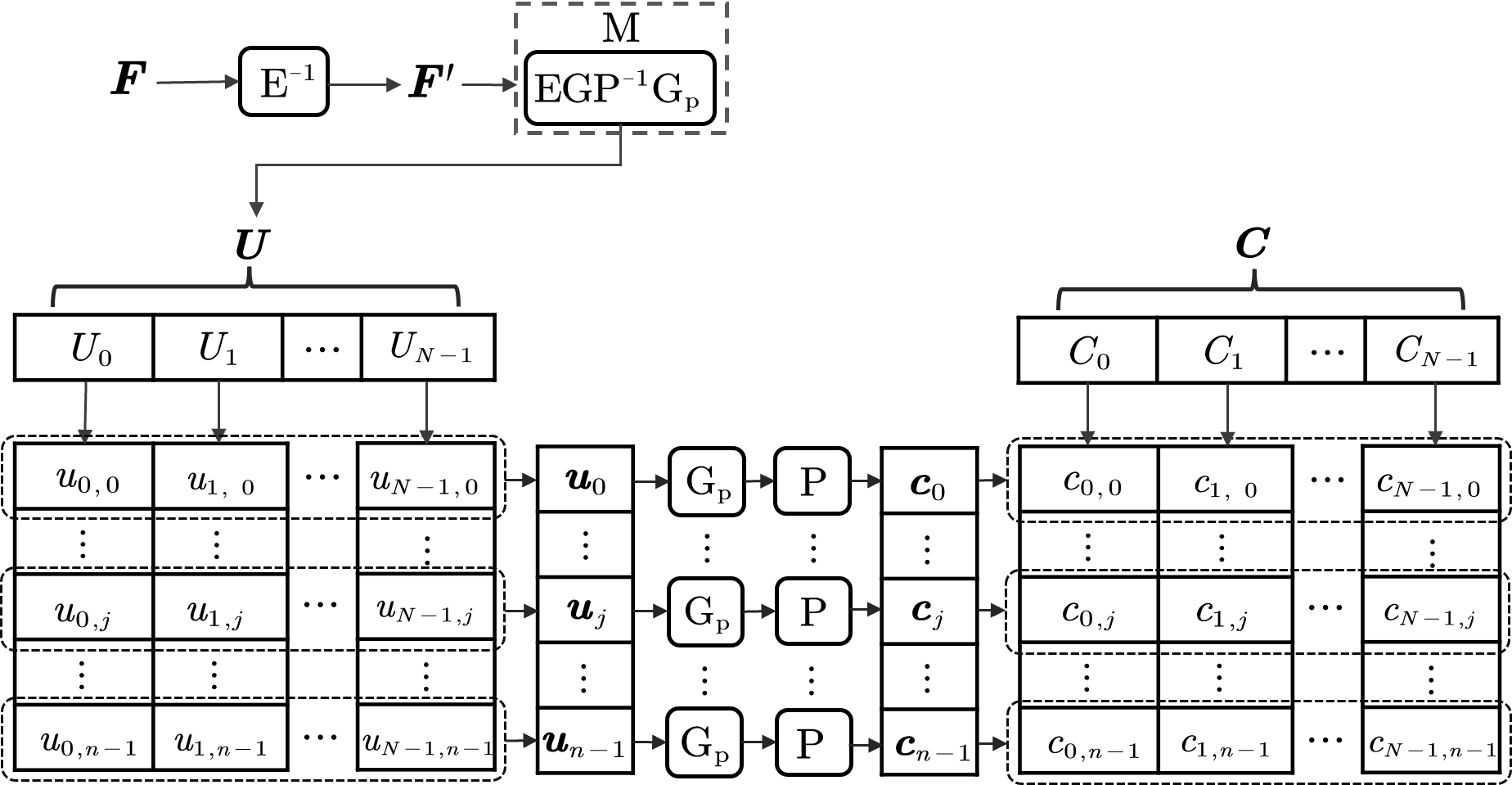}}
	\vspace{-2mm}
	\caption{Transformation from an eRS codeword into $n$ permuted binary polar codewords.}
	\vspace{-2mm}
	\label{encoding}
\end{figure}

\subsection{SC and SCL Decoding}

Assume that an eRS codeword $\bm{C}=(C_0, C_1, \cdots, C_{N-1})\in\mathbb{F}_{2^n}^{N}$ is transmitted over a memoryless channel and $\bm{y}=(y_0, y_1, \cdots, y_{N-1})\in\mathbb{R}^N$ is the channel output. Let $\bm{C}'=\bm{C}\bm{{\rm P}}^{-1}$, $\bm{y}'=\bm{y}\bm{{\rm P}}^{-1}$, and similar to \eqref{codeword_composition}, let $\bm{c}'$ denote the binary composition of $\bm{C}'$.
Let $\bm{\mathcal{L}}$ denote the log-likelihood ratio (LLR) vector of the de-permuted channel output $\bm{y}'$, i.e., 
\begin{equation}
  \begin{aligned}
    \hspace{-0.4em}\bm{\mathcal{L}} = (\mathcal{L}_{0, 0}, \mathcal{L}_{0, 1},&\cdots, \mathcal{L}_{0, n-1}, \mathcal{L}_{1, 0}, \mathcal{L}_{1, 1}, \cdots, \mathcal{L}_{1, n-1}, \\
    & \cdots, \mathcal{L}_{N-1, 0}, \mathcal{L}_{N-1, 1}, \cdots, \mathcal{L}_{N-1, n-1}).\label{received_LLRs}
  \end{aligned}
\end{equation}
Similar to \eqref{all_element_u}-\eqref{part_ele_u}, it can be partitioned into $n$ groups, each of which forms the input of an SC decoder, i.e., for $j=0, 1, \cdots, n-1$,
\begin{equation}
  \bm{\mathcal{L}}_j = (\mathcal{L}_{0, j}, \mathcal{L}_{1, j}, \cdots, \mathcal{L}_{N-1, j}).\label{partition_LLRs} 
\end{equation}
Entries of \eqref{partition_LLRs} are defined as 
\begin{equation}
  \mathcal{L}_{i, j} = \ln\frac{p(y_{i}'|c_{i, j}'=0)}{p(y_{i}'|c_{i, j}'=1)}.
\end{equation}
Note that $p(y_{i}'|c_{i, j}')$ is the channel transition probability, and $c_{i, j}'\in \mathbb{F}_2$. The SC decoding performs a greedy search and estimation of information bits over the binary tree of $n+1$ layers and $N$ leaves. Let $s$ indexes the layers of the tree, where $0\leq s\leq n$. The LLR vector of layer-$s$ can be written as 
\begin{equation}
  \bm{\mathcal{L}}_j^{(s)}=(\mathcal{L}^{(s)}_{0, j}, \mathcal{L}^{(s)}_{1, j}, \cdots, \mathcal{L}^{(s)}_{N-1, j}).
\end{equation}
Note that $\bm{\mathcal{L}}_{j}^{(n)}=\bm{\mathcal{L}}_{j}$. For $0\leq s\leq n-1$, the layer-$s$ LLRs can be computed by \cite{LLRs_update}
\begin{equation}
  \mathcal{L}^{(s)}_{i ,j}= {\rm{sign}}(\mathcal{L} ^{(s+1)}_{i, j}\cdot \mathcal{L}^{(s+1)}_{i+2^s, j})\min(|\mathcal{L}^{(s+1)}_{i, j}|, |\mathcal{L}^{(s+1)}_{i+2^s, j}|), 
\end{equation}
\begin{equation}  
  \mathcal{L}^{(s)}_{i+2^s, j} = (1-2\hat{u}^{(s)}_{i, j})\cdot \mathcal{L}^{(s+1)}_{i, j}+\mathcal{L}^{(s+1)}_{i+2^s, j}. \label{g_function}
\end{equation}
When reaching layer-$0$, estimation of information bits $\hat{u}_{i, j}^{(0)}$ are determined by their LLRs $\mathcal{L}_{i, j}^{(0)}$. In particular, for $i\in \mathcal{A}$,
\begin{equation}
  \hat{u}_{i, j}^{(0)} = \begin{cases}
  0, & \text{if } \mathcal{L}_{i, j}^{(0)} \geq 0; \\
  1, & \text{otherwise}.
  \end{cases}
\end{equation}
For $i\in\mathcal{A}^{c}$,
\begin{equation}
 \hat{u}_{i, j}^{(0)} = (\sum_{t=1}^{\tau_i}\hat{F}'_t\cdot \mathbf{M}_{t, i})\left[j\right].
\end{equation}
Estimation of non-binary input symbols are further determined by 
\begin{equation}
  \hat{U}_i = \begin{cases}
  \sum_{j=0}^{n-1}(\hat{u}^{(0)}_{i, j}\alpha^j), &\text{if } i\in\mathcal{A}, \\
  \sum_{t=0}^{\tau_i}\hat{U}^{\mathcal{A}}_t\cdot\mathbf{M}_{t, i}, &\text{if } i\in\mathcal{A}^c.
  \end{cases}
\end{equation} 
When obtaining $\hat{u}_{i, j}^{(0)}$, for $1\leq s\leq n$, the layer-$s$ estimations can be determined by 
\begin{equation}
  \begin{aligned}
    &\hat{u}_{i, j}^{(s)} = \hat{u}^{(s-1)}_{i+2^{(s-1)}, j} \oplus \hat{u}^{(s-1)}_{i, j}, \\
    &\hat{u}_{i+2^{(s-1)}, j}^{(s)} = \hat{u}_{i+2^{(s-1)}, j}^{(s-1)}. \label{bits_update}
  \end{aligned}
\end{equation}
Therefore, SC decoding of an eRS code can be performed by deploying $n$ SC decoders, each of which takes an input of $\bm{\mathcal{L}}_0, \bm{\mathcal{L}}_1, \cdots, \bm{\mathcal{L}}_{n-1}$, respectively. Vector $\hat{\bm{U}}=(\hat{U}_0, \hat{U}_1, \cdots, \hat{U}_{N-1})$ is estimated in a symbol-by-symbol manner.

Unlike the SC decoder which employs hard-decision for all $n$ bits of each information symbol in $\bm{U}$, the SCL decoder inspects both options. For each information symbol, a decoding path is split into $2^n$ paths. Similar to the SCL decoding of polar codes, in order to contain its exponentially increasing decoding complexity, only the $L$ most likely paths are reserved\cite{SCL_decoding}. 

Let $\Phi_{i}\,(l)$ denote a metric of the $l$-th path when decoding the $i$-th symbol. 
Let us assume that there are \textit{L} surviving paths after estimating $\hat{U}_{i-1}$. For $0\leq l\leq L-1$, their path metrics can be computed by 
\begin{equation}
  \Phi_{i-1}(l) = \sum_{i'=0}^{i-1}\sum_{j\in\Lambda(i', l)}|\mathcal{L}^{(0)}_{i', j}(l)|,
\end{equation}
where 
$\Lambda(i', l)=\left\{j\,|\,{\rm sign}(\mathcal{L}^{(0)}_{i', j}(l))\neq 1-2{\hat{u}_{i', j}^{(0)}}\right\}$. 
Note that $\Phi_{i-1}(l)$ measures the likelihood of path-$l$ after estimating $\hat{U}_{i-1}$, where a smaller value indicates the path is more reliable.
If $i\in\mathcal{A}$, each surviving path is split into $2^n$ paths with $\hat{U}_i$ estimated as $0, 1, \alpha, \cdots, \alpha^{2^n-2}$, respectively. The path metrics can be computed by
\begin{equation}
  \Phi_{i}(l, \hat{U}_i) = \Phi_{i-1}(l) + \sum_{j\in\Lambda(i, l)}|\mathcal{L}^{(0)}_{i, j}(l)|.
\end{equation}
Hence, at layer-$0$, $2^nL$ candidate paths are generated. Among them, only the $L$ paths with the smallest metrics are reserved. 
When $i=N-1$, the decoding path with the smallest path metric will be selected as the decoding output. For length-$N$ eRS codes, its SCL decoding that reserves $L$ paths yields a complexity of $\mathcal{O}(nLN\log_2N)$.
 
\section{SC Decoding Performance Analysis}
This section analyzes the SC decoding performance of eRS codes, shedding insights into rationale of the decoding and design of permutation matrix $\bm{{\rm P}}$.
Note that theoretical characterization of the SCL decoding performance for eRS codes remains challenging, as this is also open for polar codes. 

For an $(N, K)$ eRS code, its SC decoding error probability can be determined by\cite{polar_subcodes}
\begin{equation}
  P_{{\rm e}}^{{\rm SC}}=1-\prod_{i\in\mathcal{A}}(1-P_{\rm e}(\mathcal{W}_i))^n, \label{eRS_SCdecoding_theoretical}
\end{equation}
where $P_{\rm e}(\mathcal{W}_i)$ is the error probability of the $i$-th polarized subchannel $\mathcal{W}_i$ of a length-$N$ polar code. Note that $P_{\rm e}(\mathcal{W}_i)$ can be determined by Monte-Carlo simulation or Gaussian approxiation (GA)\cite{GA_error_prob}. 
Based on the above mentioned transformation of eRS codes, we first study the linearly independent columns in the pre-transformed matrix $\mathbf{M}$, which indicates that the information bits distribution of the polar encoders input is restrained by the generator matrix of the eRS code. Further based on the property of channel degradation, lower bound of SC decoding performance can be determined. Let us define $a=\lceil -\log_2 R \rceil$, and $\mathcal{D}=\{2^a-1, 2\cdot 2^a-1, \cdots, 2^n-1\}$, where $R=K/N$ is rate of the eRS code, and $|\mathcal{D}|=2^{n-a}$. Given matrix $\mathbf{M}$, let us denote $\mathbf{M}^{\mathcal{D}}$ as a submatrix consisting of the columns of $\mathbf{M}$ that are indexed by $\mathcal{D}$.

\subsection{Linearly Independent Columns of Pre-transformed Matrix}
In order to determine the information symbols distribution in $\bm{U}$, the pre-transformed matrix $\mathbf{M}$ needs to be analyzed.
\begin{lemma}
Given an $(N, K)$ eRS code for all $\bm{{\rm P}}\in\mathbb{F}_2^{N\times N}$, columns of $\mathbf{M}^{\mathcal{D}}$ are linearly independent. \label{linear independence of M}
\end{lemma}
\begin{proof}
  By defining $\bm{{\rm I}}_{2^{n-a}}$ and $\bm{e}_{2^a}$ as an identity matrix of dimension $2^{n-a}$ and a length-$2^a$ weight-1 binary column vector $\left[0, 0, \cdots, 0, 1\right]^{{\rm T}}$, respectively,
  $\mathbf{M}^\mathcal{D}$ can be represented as 
  \begin{equation}
    \mathbf{M}^{\mathcal{D}}=\mathbf{M}(\bm{{\rm I}}_{2^{n-a}}\otimes \bm{e}_{2^a}),
  \end{equation}
  where again $\otimes$ denotes the Kronecker product. Hence,
  \begin{equation}
    \begin{aligned}
      \rm{rank}(\mathbf{M}^\mathcal{D})&={\rm rank}(\mathbf{M}(\bm{{\rm I}}_{2^{n-a}}\otimes \bm{e}_{2^a}))\\
      &= {\rm rank}(\bm{{\rm EGP}}^{-1}\bm{{\rm G}}_{\rm p}(\bm{{\rm I}}_{2^{n-a}}\otimes \bm{e}_{2^a})).
    \end{aligned}
  \end{equation}
  Since $\bm{{\rm G}}_{\rm p} = \bm{{\rm F}}^{\otimes n}$,
  \begin{equation}
    \begin{aligned}
      \bm{{\rm G}}_{\rm p}(\bm{{\rm I}}_{2^{n-a}} \otimes \bm{e}_{2^a})
      &=(\mathbf{F}^{\otimes {n-a}}\otimes \mathbf{F}^{\otimes a})(\bm{{\rm I}}_{2^{n-a}}\otimes \bm{e}_{2^a}) \\
      &=(\mathbf{F}^{\otimes n-a}\bm{{\rm I}}_{2^{n-a}})\otimes(\mathbf{F}^{\otimes a}\bm{e}_{2^a})\\
      &\stackrel{(\star)}=(\bm{{\rm I}}_{2^{n-a}}\mathbf{F}^{\otimes n-a})\otimes \bm{e}_{2^a}\\
      &=(\bm{{\rm I}}_{2^{n-a}}\otimes \bm{e}_{2^a})\mathbf{F}^{\otimes n-a},
    \end{aligned}
  \end{equation}
  and
  \begin{equation}
    \begin{aligned}
      {\mathrm{rank}}(\mathbf{M}^\mathcal{D})
      &={\mathrm{rank}}(\bm{{\rm EGP}}^{-1}(\bm{{\rm I}}_{2^{n-a}}\otimes \bm{e}_{2^a})\mathbf{F}^{\otimes n-a})\\
      &={\mathrm{rank}}(\bm{{\rm E}}(\bm{{\rm GP}}^{-1})^{\mathcal{D}}\mathbf{F}^{\otimes n-a}). 
    \end{aligned}
  \end{equation}
  The above equality $(\star)$ is elaborated by the fact that the last column of $\mathbf{F}^{\otimes a}$ is $\bm{e}_{2^a}$, and $\mathbf{F}^{\otimes a}\bm{e}_{2^a}=\bm{e}_{2^a}$. Similarly, $(\bm{{\rm GP}}^{-1})^{\mathcal{D}}$ denotes a submatrix of $(\bm{{\rm GP}}^{-1})$, which is formed by columns indexed by $\mathcal{D}$.
  Since both $\bm{{\rm E}}$ and $\mathbf{F}^{\otimes n-a}$ are full rank matrices, 
  \begin{equation}
    \rm{rank}(\mathbf{M}^{\mathcal{D}})=\rm{rank}((\bm{{\rm GP}}^{-1})^{\mathcal{D}}).
  \end{equation}
  Since $a=\lceil -\log_2 R \rceil$, $K\geq 2^{n-a}$. Since eRS codes are MDS codes, any $K$ columns of $\bm{{\rm G}}$ are linearly independent. Consequently,
  \begin{equation}
    \mathrm{rank}(\mathbf{M}^{\mathcal{D}}) = 2^{n-a}.
  \end{equation}  
  Therefore, all columns of $\mathbf{M}^{\mathcal{D}}$ are linearly independent.
\end{proof}

Therefore, in $\mathbf{M}$, columns $2^a-1$, $2\cdot 2^a-1$, $\cdots$, $2^n-1$ are linearly independent. Since the static frozen symbols in $\bm{U}$ correspond to all-zero columns in $\mathbf{M}$, the positions indexed by $\mathcal{D}$ should carry either information symbols or dynamic frozen symbols. 

\subsection{SC Decoding Performance Lower Bound}
Based on the above conclusion, an optimal distribution of information bits in polar paradigm can be derived. Based on the property of channel degradation, we first analyze the error probabilities of the subchannels indexed in $\mathcal{D}$. 
\begin{lemma}
  \cite{channel_degradation} For length-$2^n$ polar codes with $a\leq n$, let $1\leq\theta\leq 2^{n-a}$ and $1\leq\delta\leq 2^a$, channel degradation implies
  \begin{equation}
    P_{\rm e}(\mathcal{W}_{\theta\cdot 2^a-1})\leq P_{\rm e}(\mathcal{W}_{\theta\cdot 2^a-\delta}).
  \end{equation} \label{channel degradation}
\end{lemma}
\vspace{-1.5em} 

It indicates when the $N$ polarized subchannels are divided into consecutive, non-overlapping blocks of length $2^a$, the last subchannel of each block has the lowest error probability.
\begin{theorem}
  For an $(N, K)$ eRS code, its SC decoding performance is lower bounded by
  \begin{equation}
    P_{\rm e}^{{\rm SC}}\geq 1-\prod_{i\in \mathcal{D}}(1-P_{\rm e}(\mathcal{W}_i))^{n}. \label{lower_bound}
  \end{equation} \label{SC lower bound}
\end{theorem}
\vspace{-0.7em} 

\begin{proof}
  The proof can be categorized into two cases, based on rate of the eRS code.
  
  Case \Rmnum{1}: When $a=-\log_2 R\in\mathbb{Z}^{+}$, we have $K=|\mathcal{D}|$. 
  
  Since columns of $\mathbf{M}^{\mathcal{D}}$ are linearly independent, they correspond to the positions of $\bm{U}$ that carry either information symbols or dynamic frozen symbols. Meanwhile, the polarized subchannels indexed by $\mathcal{D}$ have the lowest error probabilities in each length-$2^a$ block. Therefore, the optimal distribution of information bits can be determined as follows. Since the first information symbol is located before the first dynamic frozen symbol, the first information symbol should be at the subchannel indexed by $2^a-1$. For the subchannel indexed by $2\cdot 2^a-1$, if it carries a dynamic frozen symbol, it is a linear combination of the previous information symbols. It implies that the second information symbol should be placed ahead of it to preserve the linear independence property. Therefore, assigning the second information symbol at $2\cdot 2^a-1$ improves the SC decoding performance. This argument can be extended to the remaining indices in $\mathcal{D}$. Hence, when $K$ information symbols are placed in the positions indexed by $\mathcal{D}$, i.e., $K$ information bits occupy these positions in each polar input vector, the SC decoding performance is lower bounded by  
  \begin{equation}
     P_{\rm e}^{{\rm SC}}\geq 1-\prod_{i\in \mathcal{D}}(1-P_{\rm e}(\mathcal{W}_i))^{n}. 
  \end{equation}

  Case \Rmnum{2}: When $-\log_2R\notin\mathbb{Z}^{+}$, we have $K>|\mathcal{D}|$. Let $\mathcal{D}^*\subset\{0, 1, \cdots, N-1\}\setminus\mathcal{D}$ further denote the index set of the polarized subchannels with the lowest error probabilities, where $|\mathcal{D}^*|=K-|\mathcal{D}|$.  We have
  \begin{equation}
    \begin{aligned}
      P_{\rm e}^{\mathrm{SC}}
      &\geq 1-\prod_{i\in\mathcal{D}}(1-P_{\rm e}(\mathcal{W}_i))^n\prod_{i\in\mathcal{D}^*}(1-P_{\rm e}(\mathcal{W}_i))^n\\
      &>1-\prod_{i\in\mathcal{D}}(1-P_{\rm e}(\mathcal{W}_i))^n.
    \end{aligned}
    \vspace{-3mm}
  \end{equation}
\end{proof}

Theorem 3 shows the SC decoding performance is lower bounded by \eqref{lower_bound}, in which the information symbols distribution of $\bm{U}$ is always optimal. 

\section{Numerical study}
This section shows our numerical study on SC and SCL decoding performances for eRS codes. Our simulations are conducted over the additive white Gaussian noise (AWGN) channel with binary phase shift keying (BPSK) modulation. Performance of the Chase-BM and KV algorithms are provided as comparison benchmarks. They are parameterized by $\gamma$ that denotes the number of flipped positions and $L$ that denotes the maximum decoding output list size, respectively. The signal-to-noise ratio (SNR) is defined as $\frac{E_b}{N_0}$, where $E_b$ is the transmitted energy per information bit, and $N_0$ is the noise variance.
\begin{table}[htbp]
  \vspace{-3mm}
  \caption{SC decoding performance lower bounds for eRS codes of $N$ = 16, 32, 64, 128, 256, and $R$ = 0.25, 0.50.}
  \vspace{-2mm}
  \centering
  \renewcommand{\arraystretch}{1.1}
  \begin{tabular}{cccccc}
  \toprule  
  Rate                  & Length & FER & Rate                  & Length & FER \\ 
  \midrule
  \multirow{5}{*}{0.25} & 16     &  $1.70\times 10^{-4} $  & \multirow{5}{*}{0.50} & 16     &   $5.77\times 10^{-5}$  \\
                        & 32     &  $2.20\times 10^{-3}$   &                       & 32     &   $2.72\times 10^{-4}$  \\
                        & 64     &  $2.31\times 10^{-2}$  &                       & 64     &    $1.20\times 10^{-3}$ \\
                        & 128    &   $1.64\times 10^{-1}$  &                       & 128    &   $5.30\times 10^{-3}$  \\
                        & 256    &   $6.09\times 10^{-1}$  &                       & 256    &    $2.26\times 10^{-2}$ \\ 
  \bottomrule  
  \end{tabular} \label{same_rate_table}
  \vspace{-3mm}
\end{table}

Table \ref{same_rate_table} shows the SC decoding performance lower bounds for eRS codes of various lengths and rates, over the AWGN channel with an SNR of 11 dB. For an eRS code of rate 0.50, information bits of each polar input are transmitted over the subchannels indexed by $\left\{1, 3, \cdots, N-3, N-1\right\}$, while for the rate 0.25 eRS codes, they are transmitted over subchannels indexed by $\left\{3, 7, \cdots, N-5, N-1\right\}$. It can be seen that the SC decoding performance degrades as the codeword length increases. In constrast to the behavior of an algebraic decoding that largely depends on the code's minimum Hamming distance, the SC and SCL decoding rely on the richness of \textit{apriori} information available when estimating information symbols.

\begin{figure}[htbp]
  \vspace{-4mm}
  \centerline{\includegraphics[scale=0.41, clip, trim=10mm 0mm 10mm 7mm]{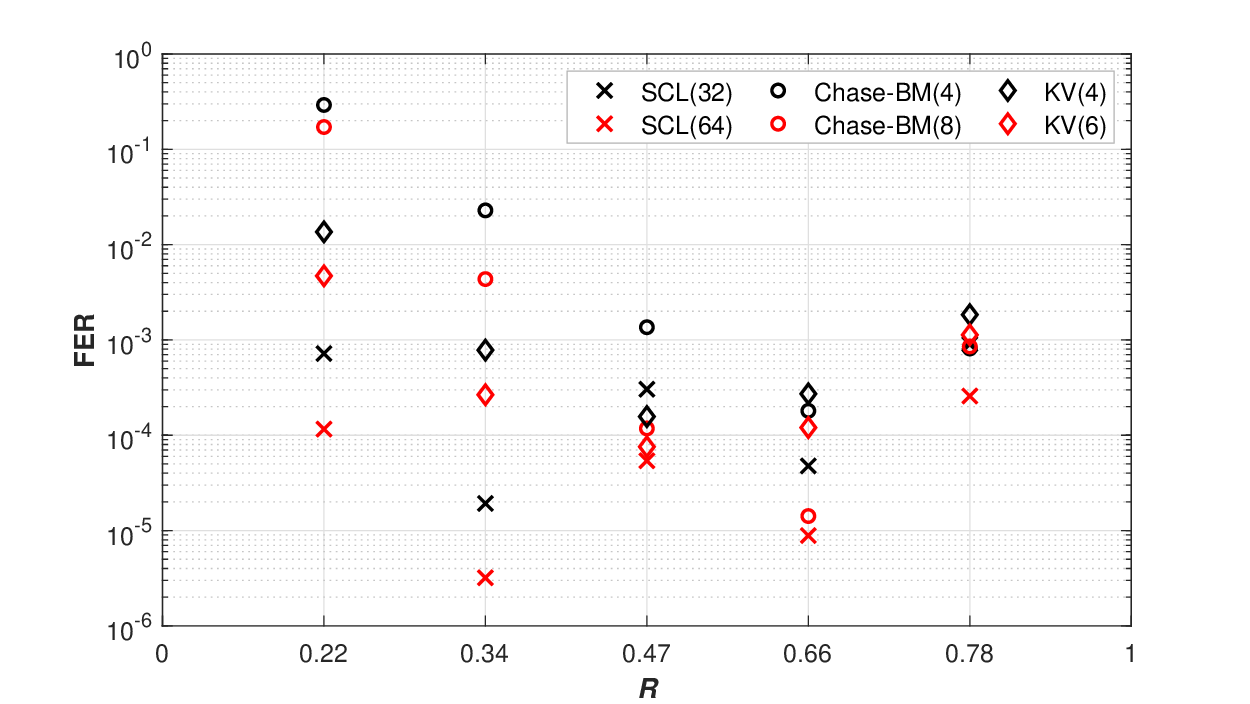}}
	\vspace{-4mm}
	\caption{SCL Decoding Performance of Length-32 eRS Codes.}
	\vspace{-8mm}
	\label{multi_algorithms_performance}
\end{figure}

\begin{table}[htbp]
	\caption{SCL Decoding complexity comparison of the decoding schemes}
	\vspace{-2mm}
	\centering
	\renewcommand{\arraystretch}{1.1}
	\label{complexity}
	\footnotesize
  \begin{tabular}{cccc}
  \toprule
  $(N, K)$            & Scheme      & $\mathbb{F}_{2^5}$ oper. & FLOPs \\ 
  \midrule
  \multirow{3}{*}{(32, 7)}& SCL(64)  & $5.23\times 10^3$ & $5.77\times 10^4$    \\
                      & Chase-BM(8) & $9.73\times 10^5$ &  $1.30\times 10^4$  \\
                       & KV(6)       & $6.36\times 10^5$ & $1.05\times 10^5$  \\ 
  \midrule
  \multirow{3}{*}{(32, 15)} & SCL(64)     &  $1.20\times 10^{4}$ &   $8.90\times 10^{4}$    \\
                      & Chase-BM(8) & $8.07\times 10^{5}$  &   $1.15\times 10^{4}$    \\
                      & KV(6)       &  $2.05\times 10^{6}$ &   $1.59\times 10^{5}$    \\
  \midrule
  \multirow{3}{*}{(32, 25)}& SCL(64)  &   $5.35\times 10^3$ & $8.05\times 10^4$\\
                      & Chase-BM(8) & $1.72\times 10^{5} $ &  $8.77\times 10^3$ \\
                       & KV(6)       & $3.05\times 10^6$ & $2.07\times 10^5$ \\ 
  \bottomrule
  \end{tabular}\label{complexity of three schemes}
  \vspace{-3mm}
\end{table}

Fig. \ref{multi_algorithms_performance} compares the SCL decoding frame error rate (FER) performance of length-32 eRS codes with rates 0.22, 0.34, 0.47, 0.66 and 0.78. They are the (32, 7), (32, 11), (32, 15), (32, 21) and (32, 25) eRS codes, respectively. The decoding FERs are obtained at the SNR of 6 dB. For the SCL decoding, this paper adopts the permutation matrix proposed by \cite{peter_SC_4_RS_codes, binary_transformation}. It produces an information symbols distribution achieving the optimal distribution specified in Case \Rmnum{1} of Theorem 3, which is the best attainable result so far. Meanwhile, Table \ref{complexity of three schemes} shows decoding complexity of these algorithms under the specified parameters. Together with Fig. \ref{multi_algorithms_performance}, they show that the SCL decoding can outperform both the KV and Chase-BM decoding with fewer finite field arithmetic operations. However, compared with Chase-BM decoding, SCL decoding requires more floating point operations (FLOPs) due to its soft decoding nature. When decoding the (32, 15) eRS codes, the SCL\,(64) decoding slightly outperforms the KV\,(6) decoding while reduces both the number of $\mathbb{F}_{2^5}$ operations and FLOPs. Compared with the Chase-BM\,(8) decoding, the SCL\,(64) decoding achieves better FER performance at the cost of increased FLOPs. However, this research shows the SCL decoding competency so far only exhibits for short-to-medium length eRS codes. As the codeword length increases, the amount of \textit{apriori} information provided for the decoding becomes less significant for the decoding to thrive.

\bibliographystyle{IEEEtranetal}
\bibliography{MyLibrary}

@Article{arikan2009polar,
	author  = {E. Arikan},
	journal = {IEEE Trans. Inf. Theory},
	title   = {Channel polarization: A method for constructing capacity-achieving codes for symmetric binary-input memoryless channels},
	year    = {2009},
	number  = {7},
	pages   = {3051-3073},
	volume  = {55},
}

@article{SCL_NonBinary,
  title={{SCL} Decoding of Non-Binary Linear Block Codes},
  author={J. Lin and L. Chen and X. Ye},
  journal={arXiv preprint arXiv:2511.11256},
  year={2025}
}

@article{binary_transformation,
    author={Lin, Chien-Ying and Huang, Yu-Chih and Shieh, Shin-Lin and Chen, Po-Ning},
    journal={IEEE Open J. Commun. Soc.}, 
    title={Transformation of Binary Linear Block Codes to Polar Codes With Dynamic Frozen}, 
    year={2020},
    volume={1},
    number={},
    pages={333-341},
    doi={10.1109/OJCOMS.2020.2979529}
}

@inproceedings{trifonov_dynamic,
  author={P. Trifonov and V. Miloslavskaya},
	booktitle={Proc. IEEE Inf. Theory Workshop (ITW)}, 
	title={Polar codes with dynamic frozen symbols and their decoding by directed search}, 
	year={2013},
  month = {Sep.},
  address = {Seville, Spain},
	volume={},
	number={},
	pages={1-5},
	doi={10.1109/ITW.2013.6691213}
}

@article{eRS_MDS,
  title={On {MDS} extensions of generalized {R}eed-{S}olomon codes},
  author={G. Seroussi and R.M. Roth},
  journal={IEEE Trans. Inf. Theory},
  volume={32},
  number={3},
  pages={349--354},
  year={2003},
  publisher={IEEE}
}

@article{GA_error_prob,
  author={P. Trifonov},
	journal={IEEE Trans. Commun.},
	title={Efficient Design and Decoding of Polar Codes},
	year={2012},
	volume={60},
	number={11},
	pages={3221-3227},
	doi={10.1109/TCOMM.2012.081512.110872},
	ISSN={1558-0857},
	month={Nov.},
}

@article{polar_subcodes,
  title={Polar subcodes},
  author={P. Trifonov and V. Miloslavskaya},
  journal={IEEE J. Sel. Areas Commun.},
  volume={34},
  number={2},
  pages={254--266},
  year={2015},
  publisher={IEEE}
}

@article{channel_degradation,
  title={Performance of polar codes with the construction using density evolution},
  author={R. Mori and T. Tanaka},
  journal={IEEE Commun. Lett.},
  volume={13},
  number={7},
  pages={519--521},
  year={2009},
  publisher={IEEE}
}

@article{LLRs_update,
  title={{LLR}-based successive cancellation list decoding of polar codes},
  author={A. Balatsoukas-Stimming and M. {Bastani Parizi} and A. Burg},
  journal={IEEE Trans. Signal Process.},
  volume={63},
  number={19},
  pages={5165--5179},
  year={2015},
  publisher={IEEE}
}

@ARTICLE{SCL_decoding,
  title={List decoding of polar codes},
	author={{Tal}, Ido and {Vardy}, Alexander},
	journal={{IEEE} Trans. Inf. Theory},
	volume={61},
	number={5},
	pages={2213--2226},
	year={2015},
	publisher={IEEE}
}

@inproceedings{peter_SC_4_RS_codes,
  author={P. Trifonov},
  booktitle ={Proc. IEEE Inf. Theory Workshop (ITW)}, 
  title={Successive cancellation permutation decoding of {R}eed-{S}olomon codes}, 
  year={2014},
  month = {Nov.},
  address = {Hobart, Australia},
  volume={},
  number={},
  pages={386-390}
}

@book{BM_algorithm,
  title={Algebraic Coding Theory},
  author={E. R. Berlekamp},
  year={1968},
  publisher={New York: McGraw-Hill}
}

@ARTICLE{GS_algorithm,
  author={V. Guruswami. and M. Sudan},
  journal={IEEE Trans. Inf. Theory}, 
  title={Improved decoding of {R}eed-{S}olomon and algebraic-geometry codes}, 
  year={1999},
  volume={45},
  number={6},
  pages={1757-1767},
}

@article{shift_register_BCH,
  author={J. Massey},
  title={Shift-register synthesis and {BCH} decoding},
  journal={IEEE Trans. Inf. Theory},
  volume={15},
  number={1},
  year = {2003},
  pages={122--127},
  publisher={IEEE},
}

@article{KV_algorithm,
  title={Algebraic soft-decision decoding of {R}eed-{S}olomon codes},
  author={R. Koetter and A. Vardy},
  journal={IEEE Trans. Inf. Theory},
  volume={49},
  number={11},
  pages={2809--2825},
  year={2003},
  publisher={IEEE}
}

@article{KV_compliexity,
  title={Progressive algebraic soft-decision decoding of {R}eed-{S}olomon codes using module minimization},
  author={J. Xing and L. Chen and M. Bossert},
  journal={IEEE Trans. Commun.},
  volume={67},
  number={11},
  pages={7379--7391},
  year={2019},
  publisher={IEEE}
}

@article{chase,
  title={Class of algorithms for decoding block codes with channel measurement information},
  author={D. Chase},
  journal={IEEE Trans. Inf. Theory},
  volume={18},
  number={1},
  pages={170--182},
  year={2003},
  publisher={IEEE}
}

@book{eRS_MDS_book,
  title={Algebraic Codes for Data Transmission},
  author={Blahut, Richard E},
  year={2003},
  publisher={Cambridge University Press}
}
\end{document}